\documentclass[10pt,article,oneside]{memoir} 
\usepackage{pdfsync}
\usepackage[utf8]{inputenc}
\usepackage[T1]{fontenc}

\usepackage{lipsum}
\usepackage{multicol}

\usepackage{amsmath, amssymb, amsthm, mathrsfs,nicefrac,MnSymbol}
\usepackage{multirow,tabularx,colortbl,hhline}
\usepackage{nicefrac}

\usepackage{listings}
\usepackage{color}
\usepackage{footnote}

\definecolor{dkgreen}{rgb}{0,0.6,0}
\definecolor{gray}{rgb}{0.5,0.5,0.5}
\definecolor{mauve}{rgb}{0.58,0,0.82}

\RequirePackage{etex}

\usepackage{graphicx,xcolor}
\graphicspath{{./figures/}}

\usepackage{framed}
\definecolor{shadecolor}{rgb}{0.9,0.9,0.9}
\usepackage{watermark}

\newtheorem{thm}{Theorem}[]
\newtheorem*{thm*}{Theorem}

\newtheorem{lem}[]{Lemma}

\newtheorem{protocol}[]{Protocol}

\theoremstyle{definition}
\newtheorem{defn}[thm]{Definition}
\newtheorem*{defn*}{Definition}

\theoremstyle{remark}

\newtheorem{rem*}[]{Remark}

\usepackage{algorithm,algorithmicx}
\usepackage{algpseudocode}

\usepackage[n,advantage,operators,sets,adversary,landau,probability,notions,logic,ff,mm,primitives,events, complexity,asymptotics,keys]{cryptocode}

\newcommand{\Z}{\mathbb{Z}}

\DeclareMathOperator{\setup}{\mathsf{Setup}}

\DeclareMathOperator{\comm}{\mathsf{Com}}
\DeclareMathOperator{\open}{\mathsf{Open}}
\DeclareMathOperator{\prove}{\mathsf{Prove}}

\DeclareMathOperator{\simulate}{\mathsf{Sim}}

\usepackage[draft,textsize=small]{todonotes}   

\usepackage[bookmarks, colorlinks=false, pdftitle={The Latus PCD scheme based on Marlin}, pdfauthor={author}]{hyperref}
  \usepackage{url}

\author{
    Hab{\"o}ck, Ulrich\\
    \texttt{ulrich@horizenlabs.io}
    \and
    Garoffolo, Alberto\\
    \texttt{alberto@horizenlabs.io}
    \and
    Di Benedetto, Daniele\\
    \texttt{daniele@horizenlabs.io}
}

\begin{document}
\title{%
Darlin: Recursive proofs using Marlin
}
\date{\today
}
\maketitle

\begin{abstract}
This document\footnotemark describes \textit{Darlin}, a succinct zero-knowledge argument of knowledge based on the  Marlin SNARK \cite{Marlin} and the `dlog' polynomial commitment scheme from \cite{BootleGroth, Bulletproofs}.
Darlin addresses recursive proofs by integrating the amortization technique from Halo \cite{Halo} for the non-succinct parts of the dlog verifier, and we adapt their strategy for bivariate circuit encoding polynomials to aggregate Marlin's inner sumchecks across the nodes the recursive scheme.
We estimate the performance impact of inner sumcheck aggregation by about 30\% in a tree-like scheme of in-degree $2$, and beyond when applied to linear recursion.
\end{abstract}

\footnotetext{This paper is the full version of the previously published extended abstract `Darlin: A proof-carrying data scheme based on Marlin'.}

\begin{KeepFromToc}
  \tableofcontents
\end{KeepFromToc}

\chapter{Introduction}
\label{s:Introduction}

Succinct non-interactive arguments of knowledge (SNARKs) are the basis for secure decentralized computations, allowing to verify the correctness of a large number of operations by a single succinct and easy to verify cryptographic proof. 
Since the advent of SNARKs \cite{Groth:2010, Gennaro:2013} practical proof systems followed soon after, e.g. Pinocchio \cite{Pinocchio}, Groth16 \cite{Groth:2016}, and Groth17 \cite{Groth:2017}.  
Whereas the first SNARKs are intrinsically connected to pairings via non-standard knowledge commitments, 
proof systems from the second wave, such as Sonic \cite{Sonic}, Aurora \cite{Aurora}, Marlin \cite{Marlin}, or Plonk \cite{Plonk}, are built in a modular way on any polynomial commitment scheme.

 
To scale over large amounts of data to be processed, recursive arguments or more generally \textit{proof-carrying data (PCD) schemes} \cite{PCD, SNARKsfromPCD}, are inevitable. 
Recursive arguments verify the existence of a previous such, and their performance is largely dependent on how efficient the verifier itself is translated into an argument. 
The issue of such a translation step is that typically the assertion to be proven is arithmetized (as a relation or circuit) over a field of a different characteristic than the proof itself, and simulating the arithmetics of a `foreign' field is costly.
The most common approach to tackle the problem is using a $2$-cycle of elliptic curves \cite{AmicablePairs, Cycles}.
Such cycles are pairs of elliptic curves in which the subgroup of one curve is of the same prime order as the base field of the other. 
Applied to pairing-based SNARKs the cycle approach requires high field sizes. 
The only known cycles are based on MNT curves of low embedding degree \cite{Cycles}, and as such they demand field sizes beyond $1,000$ bit to meet a reasonable level of security  \cite{GuillevicTNFS}\footnotemark. 
\footnotetext{
\cite{Coda} uses a cycle of MNT4/MNT6 curves with $753$ bit field sizes targeting a security level of $128$ bit.
However, improvements on the towered number field sieve \cite{GuillevicTNFS, GuillevicMNT} enforce to increase the field size up to $1,000$ bits.
}%
Second wave SNARKs are not necessarily bound to pairings, hence allow to use cycles of non pairing-friendly curves \cite{Halo, Halo2}, or such in which at least one of the curves is not pairing-friendly \cite{Mina}.  
Although allowing for smaller field sizes, the use of non pairing-friendly cycles introduces another issue.
Due to a lack of better alternatives, such constructions apply (a variant of) the `dlog' polynomial commitment scheme from \cite{BootleGroth} the verifier of which is linear in the size of the circuit to be proven; a serious obstacle for efficient recursion.
In their seminal work \cite{Halo}, Bowe et al. showed how to overcome the problem of linear verifier size by a novel approach called \textit{nested amortization}. 
In nested amortization the proof system aggregates the computational `hard parts' of the verifier outside the circuit, reducing the verification of all of them to a single expensive check at the recursion end.

Since \cite{Halo} amortization schemes became an active field of research.
B\"unz et al. \cite{Buenz} gave a more modular approach to the \cite{Halo} concept of amortization (named accumulation scheme therein).
However, their approach is less performant than the one in \cite{Halo}, which directly integrates the amortization rounds into the argument system.
Boneh et al. \cite{HaloInfinite} extend the concept of amortization to \textit{private} aggregation schemes for polynomial commitments, which allow to aggregate entire opening proofs along the nodes of a PCD scheme.
An even more radical approach for rank one constraint systems (R1CS) is followed by  \cite{PrivateAggregationR1CS}, who aggregate entire R1CS solutions over recursion.
Although both approaches lead to a significant speed-up of recursion, they come at the cost of increased proof sizes for the PCD.
The private witnesses aggregated across the nodes are as large as the circuit itself. 
For a Marlin verifier it is at about $1 \text{ MiB}$ at minimum, and multiples of that in typical applications \cite{PrivateAggregationR1CS}.

\medskip
In this document we describe the \textit{Darlin} proof carrying data scheme, the recursive SNARK for a Latus sidechain of Zendoo \cite{Zendoo}, a blockchain system which supports cross-chain communication. 
Latus sidechains are highly customizable blockchains which share the same token as the Zendoo mainchain they are bootstrapped from, and the Darlin scheme is used to provide succinct proofs of correct side chain state transitions.
Darlin is based on the Marlin argument system, modified in order to handle the aggregation of both Marlin's inner sumchecks and the `dlog' hard parts. 
According to our estimates, we expect the advantage of Darlin over standard Marlin (without inner sumcheck aggregation) to be about $30\%$ when `merging' two previous proofs, at the cost of only tripling the proof size, cf. Table \ref{t:DarlinVsMarlin}. 


\begin{table}
\caption{
The impact of inner sumcheck amortization: 
Comparison of Marlin/dlog \cite{Marlin} versus Darlin for a PCD node which verifies two previous proofs, both using the Pasta curves \cite{PastaCurves}, and using the optimization techniques from \cite{Halo} to  reduce the size of the verifier in circuit. 
We ``segmentize'' dlog commitments to speed up the prover.
See Section \ref{s:Performance} for details.
As our implementation is not ready yet, the prover times are \textit{estimates} for an Amazon EC2 G4dn instance (4 Intel Xeon@2.5 GHz + 1 NVIDIA T4). 
}
\label{t:DarlinVsMarlin}
\begin{center}
\hspace*{-0.3cm}
\begin{tabular}{|l|c|c|c|c|}
\hline
\multicolumn{2}{|l|}{dlog segment size$^\dag$} &  $2^{19}$ & $2^{18}$ &$2^{17}$
\\\hline
\multirow{2}{*}{constraints}  & Marlin$^{*}$ &  $\approx 320$ k & $\approx  384$ k & $\approx 520$ k
\\
& Darlin &   $\approx 290$ k & $\approx 320$ k & $\approx 390$ k
\\\hline
\multirow{2}{*}{proof size}  &  Marlin$^{*}$ & $\approx 4.2$ kB & $\approx 4.6$ kB & $\approx 5.3$ kB
\\
 & Darlin &  $\approx 15.3$ kB & $\approx 15.7$ kB &  $\approx 16.8$ kB
\\\hline
\multirow{2}{*}{prover time}  &  Marlin$^{*}$ &  $\approx 16.5$ s & $\approx 15.9$ s & $\approx 15.38$ s
\\
 & Darlin &   $\approx 12.3$ s ($9.6$ s$^{**}$) & $\approx 11.7$ s ($9.1$ s$^{**}$) & $\approx 11.4 $  ($8.8$ s$^{**}$)
\\\hline
\multicolumn{5}{r}{$^{*}${\small Assuming R1CS density $d=2$, which is large enough in our applications.
}}
\\
\multicolumn{5}{r}{$^{**}${\small only at the two lowest levels of a proof tree, where aggregation is trivial.}}
\end{tabular}
\end{center}
\end{table}

The document is organized as follows. 
In Section \ref{s:CoboundarySumcheck} we describe a variant of the univariate sumcheck argument from \cite{Aurora, Marlin}, inspired by the grand product argument of Plonk \cite{Plonk}. 
This variant does not rely on degree bound proofs and allows a more lightweight zero-knowledge randomization. 
In Section \ref{s:Marlin} we informally describe our variant of Marlin, which besides using a slightly different matrix arithmetization applies the sumcheck argument from Section \ref{s:CoboundarySumcheck}.
In Section \ref{s:Amortization} we recapitulate the amortization strategy for the dlog hard parts, explain the aggregation of Marlin's inner sumcheck across multiple circuits, and summarize the main recursive argument.
All formal definitions and proofs are postponed to the appendix, Section \ref{s:Appendix}.


%

\chapter{Preliminary notes}

Whenever appropriate, we formulate our protocols as \textit{algebraic oracle proofs}, with oracles as an information-theoretic model for homomorphic polynomial commitments.
An algebraic oracle proof is a multi-round protocol in which the prover responds to verifier challenges with oracles for some \textit{low-degree} polynomials, receives another challenge from the verifier. 
The prover replies with some other oracles, and so on.
The verifier is allowed to query these oracles for the values of any linear combination of their polynomials at any point she chooses.
As in algebraic holographic proofs \cite{Marlin}, the verifier may access some of the inputs only via oracle queries, but we do not assume that these oracles stem from a (circuit-specific) setup phase.
Algebraic oracle proofs can be viewed as a variant of fully linear interactive oracle protocols \cite{FullyLinearPCPs}, considering an evaluation query as a linear functional of the polynomial to be queried.
However, we shall not dwell on their separate information-theoretic security properties, nor we explicitly provide the compiler which transforms these into their corresponding ones for the resulting interactive argument systems when instantiating the oracles by a polynomial commitment scheme.
Instead, the proof of our main recursive argument from Section \ref{s:FullProtocol} relies on a compiler as used implicitly in the security analysis of the batch evaluation protocol from \cite{HaloInfinite}.

\chapter{A cohomological sumcheck argument}
\label{s:CoboundarySumcheck}

Let $F$ be a finite field, $H$ be a multiplicative subgroup of order $n$, and assume that $p(X)$ is a polynomial of arbitrary degree.
The univariate \textit{sumcheck argument} from \cite{Aurora, Marlin} is an algebraic oracle proof for showing that
\[
\sum_{x\in H} p(x) = 0.
\]
The sumcheck argument is the key ingredient to Marlin's way of proving a witness polynomial satisfying the rules of a given circuit (see Section \ref{s:Marlin}).
It is based on the fact that the above sum is equal to $n$ times the constant term of the polynomial, if $p(X)$ is of \textit{reduced form}, i.e. of a degree strictly less than the domain size $|H|=n$.
Hence showing that the reduced form of $p(X)$ has constant term zero, i.e.
\begin{equation}
\label{e:Sumcheck}
p(X) = X\cdot g(X) + h(X) \cdot (X^n-1),
\end{equation}
for some polynomials $h(X)$ and $g(X)$ whereas $deg(g(X)) < n-1$, proves the claimed sum.
To convince the verifier of \eqref{e:Sumcheck} the prover provides the oracles for $p(X)$ $g(X)$ and $h(X)$, which we denote by 
\[
[p(X)], [g(X)], [h(X)],
\]
together with a proof that $deg(g(X))\leq n-1$.
In response the verifier samples a random challenge $z\sample F$ on which the oracles are queried for $p(z)$, $g(z)$, $h(z)$.
These evaluations are used to validate the identity \eqref{e:Sumcheck} at $X=z$. 
In order to obtain (honest verifier) zero-knowledge, the prover samples a random `mask' polynomial $s(X)$ of degree at least $n$ and proves that
\begin{equation}
\label{e:SumcheckZk}
\hat p(X) = p(X) + s(X)
\end{equation}
sums up to $\sigma = \sum_{z\in H} s(z)$, which  is done by an ordinary sumcheck argument for $\hat p(X) - \sigma/n$.
See \cite{Marlin} for the details.

\medskip
Our sumcheck argument carries over the grand product argument from Plonk \cite{Plonk} to the additive setting.
Instead of using the reduced form of the polynomial $p(X)$ in question, the prover shows that the additive \textit{cocycle} with respect to the group action of $\Z$ on $H$ defined by $g$,
\begin{equation}
f_p(k,X) = \sum_{i=0}^{k-1} p(g^i\cdot X)
\end{equation}  
is a \textit{coboundary}, which is characterized by the following folklore Lemma.
\begin{lem}
\label{lem:coboundary}
Let $H$ be a multiplicative subgroup of a finite field $F$ and let $g$ be a generator of $H$.
For any univariate polynomial $p(X)$ of arbitrary degree we have $\sum_{z\in H} p(z) = 0$ if and only if there exists a polynomial $U(X)$ such that 
\begin{equation}
\label{e:sumcheckBoundary}
U(g\cdot X) - U(X) =  p(X) \mod (X^n-1).
\end{equation}
\end{lem}

\begin{proof}
Suppose that $\sum_{z\in H} p(z) = 0$. 
Define $U(X)$ on $H$ by intitializing $U(g^0) = U(1)$ to any arbitrary value, and setting
\[
U(g^k) = U(1) +\sum_{i=0}^{k-1} p(g^i)
\]
for $k=1,\ldots, n-1$.
By definition $U(g^{k+1}) = U(g^k) + p(g^k)$ for all $k$, $0\leq k\leq n-2$.
The equation also holds for $k=n-1$, since the full cycle sum $\sum_{i=0}^{n-1}p(g^i) = \sum_{z\in H} p(z) $ vanishes. 
This shows that $U(g\cdot z)-U(z) = p(z)$ for all $z$ in $H$, thus any extension $U(X)$ beyond $H$ satisfies the claimed identity $U(g\cdot X)-U(X)=p(X) \bmod (X^n-1)$.
The other direction of the proof is obvious.
\end{proof}

\medskip
The main advantage of the coboundary approach is that the algebraic oracle proof for equation \eqref{e:sumcheckBoundary} allows a more lightweight zero-knowledge randomization than that of equation \eqref{e:Sumcheck}:
Since no reduced form is needed for $U(X)$,  we can simply randomize $U(X)$ by means of the vanishing polynomial of $H$, 
\begin{align}
\label{e:sumcheckRand2}
\hat U(X) &= U(X) + (c_0 + c_1\cdot X) \cdot (X^n-1),
\end{align}
with uniformly random $c_0$, $c_1\sample F$, assuming that $\hat U(X)$ is not queried beyond the sumcheck protocol.
We describe the sumcheck argument as an algebraic oracle proof for polynomials from $R= F[X]/(X^n-1)$ with the aim to prove that the prover knows an element from  $R$ which is subject to the sumcheck $\sum_{x\in H} p(X) = 0$.


\begin{protocol}[Coboundary sumcheck]
\label{p:coboundarySumcheck}
Let $H$ be a multiplicative subgroup of a finite field $F$, $g$ be a generator of $H$ having order $n$.
The prover is given $p(X)$ from $R= F[X]/(X^n-1)$ subject to $\sum_{x\in H}p(x) = 0$, and the verifier is given the oracle of a random representant $\hat p(X)= p(X) + r(X)\cdot (X^n-1)$, where $r(X)$ is sampled uniformly from the set of polynomials of degree strictly less than $b+1$.
\footnote{The bound $b\geq 0$ corresponds to the maximum number of allowed queries for $[\hat p(X)]$ beyond the sumcheck protocol.}
\begin{enumerate}
\item 
The prover $P$ computes $U(X)$ of $deg(U(X))<n$ according to the coboundary identity \eqref{e:sumcheckBoundary}.
It computes $\hat U(X)$ as in \eqref{e:sumcheckRand2},  with $c_0,c_1\sample F$, and the quotient polynomial $h(X)$
satisfying
\begin{equation}
\label{e:coboundaryZk}
\hat U(g\cdot X)- \hat U(X) = \hat p(X) + h(X)\cdot (X^n-1).
\end{equation}
$P$ then sends $[\hat U(X)]$, $[h(X)]$ to the verifier.
\item
The verifier $V$ samples a random challenge $z\sample F\setminus H$ and queries the oracles  $[\hat U(X)]$, $[h(X)]$, and $[\hat p(X)]$ for their values at $z$. 
(The oracle aborts, if $z\in H$.)  
$V$  uses these values to verify identity \eqref{e:coboundaryZk} at $X=z$, and accepts if valid. 
(Otherwise, the verifier rejects.)
\end{enumerate}
\end{protocol}




The security analysis of Protocol \ref{p:coboundarySumcheck} (applied to a specific $\hat p(X)$) is given in the course of the proof of Theorem \ref{thm:CompleteProtocol}.
As a separate algebraic oracle proof it is perfectly complete and computationally knowledge sound, assuming the size of $H$ is negligible compared to the size of the field $F$.  
It is succinct and perfectly honest verifier zero-knowledge, assuming that each the oracles $[\hat p(X)]$ is queried outside the protocol at most another $b$ times (and $[\hat U(X)]$ is not queried at all). 
The latter is an immediate consequence of the fact that the conditional distribution of 
\[
(v_1,v_2,v_3,v_4)= \big(\hat U(g\cdot z), \hat U(z), \hat p(z), h(z)\big),
\]
conditional to $z\notin H$, is uniform on the relation $\mathcal R_z = \{(v_1,v_2,v_3,v_4)\in F^4: v_1 -v_2 - v_3 = v_4\cdot (z^n-1)\}$.
If one instantiates the oracle with a computationally binding (Definition \ref{def:Binding}) and perfectly hiding (Definition \ref{def:Hiding}) polynomial commitment scheme, the opening proof of which is an argument of knowledge (Definition \ref{def:ArgumentOfKnowledge}), then the protocol is compiled into a succinct honest verifier zero-knowledge argument of knowledge.


\chapter{Coboundary Marlin}
\label{s:Marlin}


This section describes \textit{Coboundary Marlin}, a slight variant of the Marlin SNARK \cite{Marlin}.
We introduce two changes:
First, we replace Marlin's  sumcheck argument by the coboundary argument from Section \ref{s:CoboundarySumcheck}. 
Second, we\footnote{We would like to thank A. Querol for pointing out that \cite{Lunar} also choose the Lagrange kernel. 
As a consequence, our version of the lincheck is exactly the same as theirs.} make use of the Lagrange kernel
\begin{equation}
\label{e:LagrangeKernel}
L_n(X,Y) = \frac{1}{n}\cdot \frac{Y \cdot (X^n-1) - X\cdot (Y^n-1)}{X - Y}
\end{equation}
instead of the non-normalized version $R(X,Y)=\frac{X^n - Y^n}{X-Y}$. 
The Lagrange kernel shares the same key properties as $R(X,Y)$.
It can be evaluated \textit{succinctly}, and allows a \textit{practical} sumcheck representation for the bivariate circuit polynomials, as shown below.
However, we  point out that our favor for the Lagrange kernel is mainly for esthetic reasons.
Using it allows us to argue directly with the  bivariate circuit polynomials instead of a derivative in both of Marlin's sumcheck arguments as well as our aggregation strategy from Section \ref{s:Amortization}.

\section{Arithmetization}
\medskip
We assume an arithmetic circuit $\mathcal C$ over $F$ being represented by a \textit{rank-one constraint system (R1CS)}, i.e.
\begin{equation}
\label{e:R1CS}
(A\cdot y)\odot(B\cdot y) = C\cdot y,    
\end{equation}
where we assume that $A$, $B$, $C$ are $n\times n$ matrices over $F$, $\cdot$ is the vector matrix product and $\odot$ denotes the entry-wise (Hadamard) product of vectors.
The witness vector $y\in F^n$ is composed of a public part $x$ and a private part $w$, i.e.  $y=(x\|w)$.
Notice assuming quadratic matrices is no loss in generality, as the constraint system may always be padded with dummy constraints or variables. 
Moreover, we presume that $|F|-1$ is divisible by a high power of two, assuring the existence of sufficiently large multiplicative subgroups of 2-adic order.
Subgroups of such smooth order allow for a fast Fourier transform which runs in time $O(n\log(n))$, where $n$ is the order of the subgroup.
(In the sequel we call such subgroups \textit{FFT domains}.)

In Marlin the R1CS equations are expressed over the FFT domain $H=\{z\in F : z^n-1=0\}$ using Lagrange encoding.
That is, given an arbitrary enumeration $\{z_1,\ldots,z_n\}$ of $H$ a vector $y=(y_k)_{k=1}^n$ is associated with the polynomial 
\begin{equation}
\label{e:LagrangeRepresentation}
    y(X) =  \sum_k y_k \cdot L(X, z_k).
\end{equation}
In other words, $(y_k)$ is the vector of coordinates with respect to the Lagrange basis $(L(X,x_k))_k$. 
Therefore $y\in F^n$ is a solution of \eqref{e:R1CS} if and only if its associated polynomial $y(X) =  \sum_k y_k \cdot L(X, z_k)$ satisfies
\begin{align}
\label{e:QAPc}
y_A(X)\cdot y_B(X) & = \sum_{z\in H}  C(X,z)\cdot y(z) \mod (X^n-1),
\\
\intertext{where}
\label{e:QAPa}
y_A(X) &= \sum_{z\in H}  A(X,z)\cdot y(z) \mod (X^n-1),
\\
\label{e:QAPb}
y_B(X)  &= \sum_{z\in H}  B(X,z)\cdot y(z)  \mod (X^n-1).
\end{align}
In these equations,  $A(X,Y)$, $B(X,Y)$, $C(X,Y)$ are the bivariate polynomials with the entries of the R1CS matrices $A$, $B$, $C$ respectively as Lagrange coordinates, 
\begin{equation}
\label{e:BivariateMatrixRepresentation}
M(X,Y) =  \sum_{i,j=1}^n M_{i,j} \cdot L(X, z_i)\cdot L(Y, z_j),
\end{equation}
for $M=A,B,C$.
The double sum in \eqref{e:BivariateMatrixRepresentation} is made amenable to a univariate sumcheck argument by indexing its non-zero terms over yet another FFT domain  $K=\{ w\in F : w^m - 1 = 0\}$, again assuming the existence of a sufficiently large smooth multiplicative subgroup. 
As in Marlin, we denote by
\begin{equation}
\label{e:valANDrowANDcol}
val_M(X), row_M(X), col_M(X)\in F[X]/(X^m-1)
\end{equation} 
the polynomials of degree $<m$ which index $M$'s non-zero values, their row and column indices (the latter two regarded as points from $H$, as in \eqref{e:BivariateMatrixRepresentation}), so that   
\begin{align*}
M(X, Y) &= \sum_{w \in K } val_M(w) \cdot L(X , row_M(w)) \cdot L(Y, col_M(w)).
\end{align*}
Since   $L_n(X, z)= \frac{1}{n}\cdot \frac{z \cdot (X^n-1)}{X - z}$ whenever $z$ is from $H$, we have
\begin{multline}
\label{e:MatrixRepresentationSumcheck}
M(X, Y)
=\frac{(X^n-1)\cdot (Y^n-1)}{n}
\\
\cdot\sum_{w \in K } 
\frac{val_M(w) \cdot row_M(w)\cdot col_M(w)}{\left(X - row_M(w)\right)\cdot\left( Y - col_M(w)\right)} \mod (X^m-1).
\end{multline}
This representation, which differs slightly from \cite{Marlin}, is the one we use for the second sumcheck argument, the `inner sumcheck'.
We assume that for $M=A,B,C$, the precomputed polynomials
\begin{align}
\label{e:rowcol}
row.col_M(X) &= row_M(X)\cdot col_M(X) \mod (X^m-1),
\\
\label{e:valrowcol}
val.row.col_M(X)&= val_M(X)\cdot row_M(X)\cdot col_M(X) \mod (X^m-1),
\end{align}
regarded of degree $<m$, are also part of the verifier key.

\section{The protocol}

In Marlin, public circuit inputs $x=(x_i)$ define the Lagrange representation of the \textit{input polynomial} 
\[
x(X) = \sum_{i} x_i \cdot L_\ell(x_i, X) \in F[X]/(X^\ell - 1)
\]
over an properly sized input domain $I\leq H$ of size $\ell$, and the full circuit state polynomial $y(X)$ is combined via
\begin{equation}
y(X) =  x(X) + (X^\ell - 1 )\cdot w(X),
\end{equation}
using a gauged witness polynomial $w(X)\in F^{< n - \ell}[X]$.
The prover provides the oracles for the private witness polynomial $w(X)$, $y_A(X)$, $y_B(X)$ and convinces the verifier of the R1CS identities \eqref{e:QAPc}, \eqref{e:QAPa}, and \eqref{e:QAPb}.
These three identities are reduced to a single one by building a random linear combination based on a challenge $\eta\sample F$, i.e.
\begin{equation}
\label{e:lincheck}
y_\eta(X) = \sum_{z\in H}  T_\eta(X,z) \cdot y(z)  \mod (X^n-1),
\end{equation}
with
\begin{equation*}
y_\eta(X)  = y_A(X) + \eta \cdot y_B(X) + \eta^2\cdot y_A(X)\cdot y_B(X),
\end{equation*}
and
\begin{equation*}
T_\eta (X,Y) =A(X,Y)+ \eta\cdot B(X,Y)+ \eta^2\cdot C(X,Y).
\end{equation*}
The linear identity \eqref{e:lincheck} is reduced to a sumcheck over $H$ by sampling a polynomial $R(X,\alpha)$ using a suitable kernel $R(X,Y)$, $\alpha\sample F$, and applying it via scalar product to both sides of the equation.
This yields
\begin{equation*}
\sum_{z\in H}  \langle R(X,\alpha),T_\eta(X,z)\rangle_H \cdot y(z) 
=\langle R(X,\alpha), y_\eta(X)\rangle_H,
\end{equation*}
hence
\begin{equation}
\label{e:OuterSumcheckR}
\sum_{z\in H}  \langle R(X,\alpha),T_\eta(X,z)\rangle_H \cdot y(z) - R(z,\alpha)\cdot y_\eta(z) = 0.
\end{equation}
Choosing the Lagrange kernel $L_n(X,Y)$ for $R(X,Y)$, $\langle L_n(X,\alpha), T_\eta(X,z)\rangle_H = T_\eta(\alpha,z)$, since $T_\eta(X,z)$ is of degree less than $n$ (see Appendix \ref{s:LagrangeKernel}). 
Hence equation \eqref{e:OuterSumcheckR} is equal to
\begin{equation}
\label{e:OuterSumcheck}
\sum_{z\in H}  T_\eta(\alpha,z) \cdot y(z) -  L_n(z,\alpha)\cdot  y_\eta(z) = 0.
\end{equation}
Equation \eqref{e:OuterSumcheck} is the central identity to be proven by the protocol.


\medskip
We describe the protocol as algebraic oracle proof.

\subsubsection{Initialization}
In the first step the prover computes the polynomials\footnote{%
Unless stated otherwise we assume polynomials $p(X)$ from $F[X]/(X^n-1)$ of reduced form, i.e. of degree $<n$.
} 
$w(X)$,  $y_A(X)$, $y_B(X)\in F[X]/(X^n - 1)$ from their Lagrange representations, and chooses random representants 
\[
\hat w(X), \hat z_A(X), \hat z_B(X) \in F^{<n+1}[X] 
\]
according to the sampling rule $\hat p(X) \sample p(X) + F \cdot (X^n-1)$ using randomizer polynomials of degree zero.
It sends their oracles $[\hat w(X)]$, $[\hat y_A(X)]$, $[\hat y_B(X)]$ to the verifier,  who returns the randomnesses $\eta\sample F$ and $\alpha\sample F\setminus H$ for Equation \eqref{e:OuterSumcheck}.

\subsubsection{Outer sumcheck}
To prove equation \eqref{e:OuterSumcheck} we apply the coboundary argument from Section \ref{s:CoboundarySumcheck} to $\hat p(X) := T_\eta(\alpha, X)\cdot \hat y(X) - L_n(X,\alpha) \cdot \hat y_\eta(X)$, where 
\begin{align*}
\hat y(Y) &:= x(Y) + (Y^\ell-1)\cdot \hat w(Y), 
\\
\hat y_\eta(Y) &:=  \hat y_A(Y)+\eta\cdot \hat y_B(Y) + \eta^2\cdot \hat y_A(Y)\cdot \hat y_B(Y).
\end{align*}
The prover computes the boundary polynomial $U_1(X) \in F[X]/(X^n - 1)$, chooses a random representant 
\[
\hat U_1(X)\in  F^{<n+2}[Y]
\]
of it, and computes $h_1(X)\in F^{< 2\cdot n + 1}[X]$ for the outer sumcheck identity
\begin{multline}
\label{e:OuterSumcheckCoboundary}
T_\eta(\alpha,X) \cdot \hat y(X) -  L_n(X,\alpha)\cdot \hat y_\eta(X) 
\\
= \hat U_1(gX) - \hat U_1(X) + h_1(X)\cdot (X^n-1),
\end{multline}
where $g$ is a generator of $H$.
It then sends $[\hat U_1(X)]$, $[h_1(X)]$ together with $[T_\eta(\alpha, X)]$ to the verifier. 
The verifier samples another random challenge $\beta\sample F\setminus H$ and queries the oracles for
$\hat w(\beta), \hat y_A(\beta), \hat y_B(\beta), 
T_\eta(\alpha,\beta), \hat U_1(g\cdot \beta)$, $\hat U_1(\beta), h_1(\beta)$,
which are used for checking the identity \eqref{e:OuterSumcheckCoboundary} at $Z=\beta$.


\subsubsection{Inner sumcheck}
To  prove that  $T_\eta(\alpha,\beta)$ as provided by the oracle in fact stems from the circuit polynomials $M(X,Y)$, $M=A,B,C$ we adapt Marlin's inner sumcheck to our representation \eqref{e:MatrixRepresentationSumcheck}.
Using these we obtain
\begin{align}
\label{e:innerSumcheck}
T_\eta(\alpha, \beta)
=\sum_{w \in K }\sum_{M=A,B,C} \eta_M  \cdot \frac{val.row.col_M(w)}{(\alpha - row_M(w)) \cdot (\beta - col_M(w))},
\end{align}
where $(\eta_A,\eta_B,\eta_C)= \frac{(1-\alpha^n)\cdot (1-\beta^n)}{n^2}\cdot (1,\eta,\eta^2)$.
We apply the coboundary sumcheck to 
\[
p(X) = \sum_{M=A,B,C} \eta_M  \cdot \frac{val.row.col_M(X)}{(\alpha - row_M(X)) \cdot (\beta - col_M(X))}, 
\]
regarded as a reduced element from $F[X]/(X^m - 1)$.
The prover computes $U_2(X)$ from $F[X]/(X^m-1)$ satisfying
\begin{align*}
p(X) = \frac{T_\eta(\alpha,\beta)}{m} + U_2(g_K X) - U_2(X)  \mod (X^m-1),
\end{align*}
and then multiplies both sides with the denominator 
\begin{align*}
b(X) &= \prod_{M=A,B,C} (\alpha -row_M(X))\cdot (\beta - col_M(X))
\\
&=\prod_{M=A,B,C} \left(\alpha\beta + \beta \cdot row_M(X) + \alpha \cdot col_M(X) + row.col_M(X)\right),
\end{align*}
where $row.col_M(X)$ are the precomputed products \eqref{e:rowcol} from the prover key. 
This yields the \textit{inner sumcheck} identity
\begin{multline}
\label{e:InnerSumcheck}
\sum_{M=A,B,C} \eta_M \cdot val.row.col_M(X)
\\
= b(X) \cdot \left(\frac{T_\eta(\alpha,\beta)}{m} + U_2(g_K X) - U_2(X)\right) + h_2(X)\cdot (X^m -1 ),
\end{multline}
where $g_K$ is a generator of $K$ and $h_2(X) \in F^{< 3\cdot m - 3}[X]$.
The prover sends the oracles $[U_2(X)]$ and $[h_2(X)]$ to the verifier, who samples a random challenge $\gamma\sample F$, 
on which the oracles are queried for
$row_M(\gamma), col_M(\gamma), row.col_M(\gamma)$, $val.row.col_M(\gamma)$, where $M=A,B,C$,
and $U_2(g_K \cdot\gamma), U_2(\gamma), h_2(\gamma)$.
These values are used by the verifier to check the identity \eqref{e:InnerSumcheck} at $X=\gamma$.

\subsection{Security}

The security analysis of Coboundary Marlin is similar to that of our main recursive argument, Theorem \ref{thm:CompleteProtocol}.
As for Theorem \ref{thm:CompleteProtocol}, we stress the fact that we use the Halevi-Micali \cite{PoKHaleviMikali} notion of proof of knowledge with negligible knowledge error.
The proof can be found in the appendix, Section \ref{s:ProofCoboundaryMarlin}.
\begin{thm}
\label{thm:CoboundaryMarlin}
Instantiating the oracle by a computationally binding and perfectly hiding polynomial commitment scheme (Definition \ref{def:Hiding} and \ref{def:Binding}, Coboundary Marlin is a succinct, perfect honest verifier zero-knowledge (Defintion \ref{def:ZeroKnowledge}) argument of knowledge (Definition \ref{def:ArgumentOfKnowledge}) for the R1CS relation 
\[
\mathcal R =\big\{ ((A,B,C,x),w) : y=(x,w) \text{ satisfies } (A\cdot y)\odot(B\cdot y) = C\cdot y \big\}.
\]
\end{thm}

Using the Fiat-Shamir transform the interactive argument is transformed into a zk-SNARK with analog security properties   
in the random oracle model.

\section{A note on performance}
\label{s:NoteOnPerformance}
Marlin's outer sumcheck takes place over the FFT domain $H$, the size of which covers the number of constraints/variables of the  constraint system.
In practice circuits yield about the same number of variables as constraints, hence it is reasonable to take $n$ the number of constraints as measure for the computational effort of the outer sumcheck, assuming a sufficiently smooth order of $F^*$ to optimally match $n$.
The inner sumcheck runs over the FFT domain $K$ of size $m\approx \max_{M=A,B,C}\|M\|$ ($\|M\|$ is the number of non-zero entries in $M$), again under the assumption of sufficient smoothness.
This domain is by the factor
\[
d = \frac{\max_{M=A,B,C}\|M\|}{n}
\]
larger, where $d$ is the \textit{R1CS density} of the circuit. 
The R1CS density is the average number of variables per constraint.
In practice, we observed values between $d=1.5$ and $d=2$ for the circuits we target. 
(These circuits implement elliptic curve arithmetics over non-extension fields and the $x^5$-Poseidon hash \cite{Poseidon} with an internal state of $3$ field elements.)

\begin{table}[h!]
\caption{%
Computational effort of the (coboundary) zk-Marlin prover, using an elliptic curve based linear polynomial commitment scheme. 
We only count fast Fourier transforms $\textsf{FFT}(a)$ in terms of their domain size $a$, and elliptic curve multi scalar multiplications $\textsf{MSM}(b)$ in terms of the number of scalars $b$. 
(Without opening proof.)
}
\label{t:CoboundaryMarlin}
\vspace*{3mm}
\centering
\hspace*{-0.5cm}
\begin{tabular}{|l|c|c|}
\cline{2-3}
\multicolumn{1}{c|}{} &polynomial arithm. & commit
\\\hline
 intial round  &  $3 ~\textsf{FFT}(n)$   & $3~\textsf{MSM}(n)$
\\
%
%
\multirow{2}{*}{%
outer sumcheck}  
	&  $2 ~\textsf{FFT}(n) + 2~\textsf{FFT}(2n)$%
		& \multirow{2}{*}{$2 ~\textsf{MSM}(n) + 1~\textsf{MSM}(2n)$}
\\ & $+ 3~\textsf{FFT}(3n)$ &
\\
%
%
inner sumcheck  
	& $1~\textsf{FFT}(m) + 1~\textsf{FFT}(4m) $
		& $1~\textsf{MSM}(m) +1~\textsf{MSM}(3m)$ 

\\\hline
overall 
	 &  $\approx  (15 + 5\cdot d)~\textsf{FFT}(n)$ 
		& $\approx (7+4\cdot d)~\textsf{MSM(n)}$
\\\hline
\end{tabular}
\end{table}

\chapter{Recursion}
\label{s:Amortization}

Our recursive scheme is based on Coboundary Marlin and the \cite{Buenz} variant of the \textit{dlog polynomial commitment scheme} from \cite{BootleGroth}.
We take Coboundary Marlin without inner sumcheck as succinct argument, and we aggregate both the non-succinct parts of the opening proof verifier, as well as the correctness checks usually served by the inner sumchecks, which is verifying that the commitment intended for 
\[
T_\eta(\alpha,Y) = \sum_{M=A,B,C} \eta_M\cdot M (\alpha, Y)
\]
in fact carries these polynomials.
Aggregation of the non-succinct part of the dlog verifier (the \textit{dlog hard parts}) relies on the same principle as introduced by Halo \cite{Halo}.
The way we aggregate the inner sumchecks is a generalization of Halo's strategy for their circuit encoding polynomial $s(X,Y)$, and we  extend it across circuits to serve a reasonable number of instances $\mathcal C_i=\{A_i,B_i,C_i\}$ simultaneously.
As a separate `stand-alone' protocol, our strategy may be taken as \textit{public aggregation scheme} in the sense of \cite{HaloInfinite}, or an \textit{(atomic) accumulation scheme} according to \cite{Buenz, PrivateAggregationR1CS}.
However, for efficiency reasons we choose Halo's `interleaved' approach instead of the blackbox constructions from  \cite{Buenz, PrivateAggregationR1CS, HaloInfinite}, and let the rounds of both the argument system and the aggregation scheme share the same opening proof.

In our recursive argument certain previous proof elements  $(acc_i)_{i=1}^\ell$ called \textit{accumulators} are `passed' through inputs of the `current' circuit and post-processed within the run of the current argument.
Formally, $(acc_i)_{i=1}^\ell$ satisfy a given predicate $\phi$, 
\[
\phi(acc_i)=1, \quad i=1,\ldots,\ell,
\]  
and are mapped to dedicated inputs of the current circuit.
Beyond the rounds for proving satisfiability of the current circuit, the accumulators $(acc_i)_{i=1}^\ell$ are aggregated within some extra rounds into a new instance, the `current' accumulator $acc$, which is again subject to $\phi(acc)=1$.
Altogether our recursive argument is of the form
\[
\big\langle \prove ((acc_i)_{i=1}^\ell, (x,w), pk), \verify((acc_i)_{i=1}^\ell,x,vk) \big\rangle,
\]
where $(x,w)$ are public and private circuit witnesses, $pk$ and $vk$ are  the prover and verifier key for both Marlin and the aggregation scheme, and the new $acc$ is output to both prover and verifier.  

\section{Inner sumcheck aggregation} 
\label{s:InnerSumcheckAggregation}

Here, the accumulator consists of a commitment $C$ and the succinct description of the circuit polynomial $T_\eta(z,Y)$ intended to be represented by $C$, i.e. the point $z\in F$ and the randomnesses $\vec\eta=(\eta_A,\eta_B,\eta_C)\in F^3$, 
\[
acc_{T} = (z, \vec\eta, C).
\]
The corresponding predicate $\phi_T$ is satisfied if and only if $C$ is the commitment of $T_\eta(z,Y)$ (using commitment randomness zero).
The prover reduces the correctness of several accumulator instances to that of a single new one, and the verifier validates the correctness of this reduction while keeping track of the polynomial descriptions (i.e. the point $z$ and the coefficient vector $\vec\eta$)  by herself.
We sketch the strategy assuming a single previous accumulator.

There, a previous instance $(\alpha',\vec\eta',C')$ is `merged' with  $(\alpha,\vec\eta,C)$ of the current outer sumcheck.
In a first step, the prover reduces the  `multi-point', `multi-polynomial' instance\footnote{%
Here `multi-point' refers to the different points $\alpha$, $\alpha'$, and `multi-polynomial' to the different polynomials defined by $\vec\eta$, $\vec\eta'$.
}
$T_{\vec\eta'}(\alpha',Y)$, $T_{\vec\eta}(\alpha,Y)$ to a single-point, multi-polynomial instance 
\[
T_{\vec\eta'}(X,\beta), T_{\vec\eta}(X,\beta),
\] 
with random $\beta\sample F$, by providing the commitments to these new  polynomials and proving consistency via polynomial testing:
If the old polynomials evaluate at the challenge $\beta$ to the same values as the new polynomials at the old point, respectively, then correctness of the new polynomials overwhelmingly implies that of the old ones.
Using the same principle once again, correctness of the single-point multi-polynomial instance is then reduced in batch to a single-point single-polynomial instance 
\[
\lambda\cdot T_{\vec\eta'}(\alpha'',Y) + T_{\vec\eta}(\alpha'',Y) = T_{\lambda\cdot\vec\eta' + \vec\eta}(\alpha'',Y),
\] 
where $\lambda,\alpha''\sample F$ are random.
Note that the resulting polynomial is again of the form $T_{\vec\eta''}(\alpha'',Y)$ with $\eta''=\lambda\cdot\vec\eta' + \vec\eta$.
For the reduction, the prover shows that the linear combination $\lambda\cdot T_{\vec\eta'}(X,\beta) + T_{\vec\eta}(X,\beta)$ opens at the new challenge $X=\alpha''$ to the same value as the new polynomial  $\lambda\cdot T_{\vec\eta'}(\alpha'',Y) + T_{\vec\eta}(\alpha'',Y)$ at the old point $Y=\beta$.
Again, correctness of the new polynomial overwhelmingly implies correctness of the old ones.

Protocol \ref{p:InnerSumcheckAggregation} is regarded as a subprotocol of our complete recursive argument Protocol \ref{p:CompleteArgument}, right after the outer sumcheck.
We formulate it as an algebraic oracle protocol, considering commitments as oracles. 

\begin{protocol}[Inner sumcheck aggregation]
\label{p:InnerSumcheckAggregation}
Suppose that $acc_T'=(\alpha',\vec\eta',[T'(Y)])$ is a previous accumulator, intended to represent an oracle for $T'(Y)= T_{\vec\eta'}(\alpha',Y)$, and $(\alpha,\vec\eta, [T(Y)])$ is as provided by the prover in the current outer sumcheck, intended to represent an oracle for $T(Y)= T_{\vec\eta}(\alpha,Y)$, with $\vec\eta =(1,\eta,\eta^2)$.
Aggregation of $acc_T'$ and $(\alpha,\vec\eta,[T(Y)])$ is done according to the following steps immediately processed after the outer sumcheck.
\begin{enumerate}
\item
\label{p:InnerSumcheckAggregation:Step1}
Given $\beta$, the random challenge from the outer sumcheck, the prover sends the oracles for the `bridging polynomials'
\[
T_{\vec\eta}(X,\beta), T_{\vec\eta'}(X, \beta)\in F[X]/(X^n-1),
\]
on which the verifier responds with random $\lambda,\gamma\sample F$.

\item
\label{p:InnerSumcheckAggregation:Step2}
Given $\lambda,\gamma$ from the verifier, the prover `responds' with the oracle for
\[
T''(Y) = T_{\vec\eta}(\gamma, Y) + \lambda \cdot  T_{\vec\eta'}(\gamma, Y).
\]
\end{enumerate}
The verifier queries $[T_{\vec\eta}(X,\beta)]$, $[T_{\vec\eta'}(X,\beta)]$ for their corresponding values $v_1$, $v_2$ at $X=\alpha$ and $\alpha'$, and checks them against the values of $[T(Y)]$, $[T'(Y)]$ at $Y=\beta$, respectively.
It also queries $[T''(Y)]$ at $Y=\beta$ and checks its value against that of  the linear combination $[T_{\vec\eta}(X,\beta)]+\lambda [T_{\vec\eta'}(X,\beta)]$ at $X=\gamma$.
If these checks succeed, then the verifier accepts and the new accumulator is 
\[
acc_T''=(\alpha'',\vec\eta'', C'') = (\gamma, \vec\eta + \lambda\cdot \vec\eta', [T'(Y)]).
\]
\end{protocol}

A formal analysis of Protocol \ref{p:InnerSumcheckAggregation} is given in the course of the security proof of the complete recursive argument.
As a stand-alone argument having its own opening proof, the protocol defines a \textit{(perfectly) complete} and \textit{sound accumulation scheme} for the predicate $\phi_T$ in the sense of \cite{Buenz}: 
If both $acc_T'$ and $(\alpha,\eta,C)$ satisfy the predicate $\phi$, so does $acc_T''$.
And if $\phi(acc_T'')=1$, then with overwhelming probability both $\phi(acc_T')$ and $\phi(\alpha,\eta,C)=1$.

\section{Generalization to several circuits} 
\label{s:InnerSumcheckAggregationGeneral}

The aggregation strategy from Section \ref{s:InnerSumcheckAggregation} is easily extended to serve multiple circuits  $C_1,\ldots, \mathcal C_L$ simultaneously.
This `cross-circuit' generalization is especially useful in `non-homogeneous' chemes which are composed by a variety of recursive circuits.
Lets assume that the R1CS matrices $A_i,B_i,C_i$ of the circuits $\mathcal C_i$, $i=1,\ldots,L$, are padded to the same square dimension so that we may regard their
\[
A_i(X,Y), B_i(X,Y), C_i(X,Y),
\]
as bivariate polynomials over the same domain $H\times H$.
As in the single-circuit setting we leverage the linearity of the commitment scheme and keep track of a single \textit{cross-circuit polynomial}
\begin{equation}
\label{e:CrossCircuitT}
T_{H}(\alpha,Y)=\sum_{i=1}^L T_{i, \vec\eta_i}(z,Y)= \sum_{i=1}^L \sum_{M= A_i,B_i,C_i} \eta_{M,i}\cdot M(\alpha, Y)
\end{equation}
by means of the cross-circuit coefficient vector $H=(\vec\eta_1,\vec\eta_2, \ldots, \vec\eta_L)$.
The \textit{cross-circuit accumulator} for the collection $\mathcal C=\{\mathcal C_1,\ldots, C_L\}$ is of the form
\[
acc_\mathcal C= (\alpha, H, C),
\]
with $\alpha\in F$, coefficient vector $H=(\vec\eta_1,\ldots, \vec\eta_L)\in (F^{3})^L$, and an element $C$ from the commitment group.
The corresponding predicate $\phi_\mathcal C$ is satisfied if and only if $C$ is in fact the dlog commitment of $T_H(\alpha,Y)$, using blinding randomness zero.



\section{Accumulating the dlog hard parts}
\label{s:IPAAggregation}

The aggregation strategy for the non-succinct part of the dlog verifier is identical to that in \cite{Buenz}.
The opening proof for the dlog commitment is an inner product argument that uses the folding technique from \cite{BootleGroth} to gradually reduce the opening claim on the initial full-length polynomial to one of half the size, until ending up with the opening claim of a single coefficient polynomial. 
The final committer key $G_{f}$ of the opening proof is a single group element which is the result of a corresponding folding procedure on the full-length committer key of the dlog scheme.
It equals the commitment of the succinct \textit{reduction polynomial}
\begin{equation}
\label{e:BulletPolynomial}
h(\vec\xi, X) = \prod_{i=0}^{k-1} ( 1- \xi_{k-1-i} \cdot X^{2^i}),
\end{equation}
where $k=\log|H|=n$ is the number of reduction steps and $\vec\xi = (\xi_i)_{i=0}^{k-1}$ their challenges.
The \textit{dlog accumulator} is of the form
\[
acc_{dlog} = (\vec\xi, C),
\]
where $\vec\xi\in F^k$ and $C$ is from the commitment group, and the corresponding accumulator predicate $\phi_{dlog}$ is satisfied if and only if $C$ is the commitment of $h(\vec\xi, X)$, using blinding randomness zero.

As Protocol \ref{p:InnerSumcheckAggregation}, the aggregation strategy is regarded as a subprotocol of the complete recursive argument Protocol \ref{p:CompleteArgument}, and for efficiency reasons we reuse the challenge $\gamma$ from the inner sumcheck aggregation.
We again restrict to the case of a single previous accumulator.

\begin{protocol}[dlog hard parts aggregation]
\label{p:IPAAggregation}
Suppose that $acc_{dlog}'=(\vec\xi',[h'(X)])$ is a previous dlog accumulator, with $[h'(X)]$ representing an oracle for $h'(X)=h(\vec\xi',X)$.
The following step is part of the complete recursive argument and processed immediately after Protocol \ref{p:InnerSumcheckAggregation}:
\begin{enumerate}
\item
\label{p:InnerSumcheckAggregationStep1}
The verifier queries $[h'(X)]$ at for its value $v'$ at  $X=\gamma$ from step \eqref{p:InnerSumcheckAggregation:Step2} of Protocol \ref{p:InnerSumcheckAggregation}. 
\end{enumerate}
If $v'=h(\vec\xi', \gamma)$ then the verifier accepts.
The new accumulator $acc_{dlog}''=(\vec\xi'',C'')$ is the one from the dlog opening proof at the end of the complete protocol.
\end{protocol}


\section{The main recursive argument}
\label{s:FullProtocol}

The complete recursive argument is a composition of Coboundary Marlin's outer sumcheck for the `current' circuit, choosing `zero-knowledge bound' $b=1$, the aggregation rounds from the cross-circuit variant of Protocol \ref{p:InnerSumcheckAggregation}, and Protocol \ref{p:IPAAggregation}.
As in Section \ref{s:InnerSumcheckAggregationGeneral} we assume that the bivariate circuit polynomials $A_i(X,Y)$, $B_i(X,Y)$, $C_i(X,Y)$ are over the same domain $H\times H$, where $|H|=n$.
The query phases of these subprotocols are gathered at the end of the protocol, which is then concluded by the batch evaluation argument from \cite{HaloInfinite}.


We formulate the complete argument with oracles for polynomials replaced by their dlog commitments, while keeping with the same notation $[p(X)]$.
For simplicity, we again restrict to the case of a single previous accumulator. 
The general case is straight-forward.

\begin{protocol}[Complete recursive argument]
\label{p:CompleteArgument}
Given a composed accumulator $acc'=(acc_{\mathcal C}',acc_{dlog}')$, where $acc_\mathcal C=(\alpha', H',C_T')$ is a cross-circuit accumulator for the collection $\mathcal C=\{\mathcal C_1,\ldots,\mathcal C_L\}$ and $acc_{dlog}'=(\vec\xi, C')$ is a dlog accumulator.
The recursive argument for an instance $(x,w)$ of the `current' circuit $\mathcal C_k$ from $\mathcal C$ is composed by the following steps.
\begin{enumerate}
\item
\textit{Intitialization for $\mathcal C_k$}:
The prover computes the gauged witness polynomial $w(X)$, $z_A(X)$, and $z_B(X)$ from $F[X]/(X^n-1)$ and chooses random representants
\begin{align*}
\hat w(X), \hat z_A(X), \hat z_B(X) \in F^{<n+1}[X]
\end{align*}
as described in Section \ref{s:Marlin}.
It sends their dlog commitments $[\hat w(X)]$, $[\hat z_A(X)]$, and $[\hat z_B(X)]$ to the verifier, who responds with $\eta,\alpha\sample F$.

\item
\textit{Outer sumcheck for $\mathcal C_k$}:
The prover computes 
\[
T_{\vec\eta}(\alpha, X)= \eta_A\cdot A(\alpha,X) + \eta_B\cdot B(\alpha,X) + \eta_C\cdot C(\alpha,X) \in F[Y]/(X^n-1)
\]
of the current circuit, using $\vec\eta=(\eta_A,\eta_B,\eta_C)=(1,\eta,\eta^2)$, and
\[
\hat U_1(X)\in  F^{<n+2}[Y], h_1(X)\in F^{< 2\cdot n + 1}[X]
\]
subject to the outer sumcheck identity \eqref{e:OuterSumcheckCoboundary}.
It sends  $[T_{\vec\eta}(\alpha, X)]$, $[\hat U_1(X)]$, $[h_1(X)]$ to the verifier, who returns another random challenge $\beta\sample F$.

\item
\textit{Inner sumcheck aggregation, Step 1}:
The prover computes the `bridging' polynomials for  
\[
T_{\vec\eta}(X,\beta), T_{H'}(X,\beta)\in F[X]/(X^n-1),
\]
and sends $[T_{\vec\eta}(X,\beta)], [T_{H'}(X, \beta)]$ to the verifier, who answers with another random $\lambda,\gamma\sample F$.
\item
\textit{Inner sumcheck aggregation, Step 2}:
The prover computes the cross-circuit linear combination
\[
T_{H''}(\gamma,Y) = T_{\vec\eta}(\gamma, Y) + \lambda \cdot  T_{H'}(\gamma, Y)\in F[Y]/(Y^n-1),
\]
and $[T_{H''}(\gamma,Y)]$ to the verifier.
\end{enumerate}
After these steps, both prover and verifier engage in the batch evaluation argument from \cite{HaloInfinite} for the dlog commitment scheme, applied to the queries as listed below.
If the queried values pass the checks of the outer sumcheck, Protocol \ref{p:InnerSumcheckAggregation} and Protocol \ref{p:IPAAggregation}, and if $(acc_{\mathcal C}',acc_{dlog}')$ match with the public input $x$ of the circuit, then the verifier accepts.
The new accumulator is $acc'' =(acc_\mathcal C'', acc_{dlog}'')$ with\footnotemark 
\[
acc_\mathcal C''=(\gamma,H'', C'') = (\gamma,  \vec\eta\cdot\delta_k + \lambda \cdot H', [T_{H''}(\gamma, Y)]) ,
\] 
and $acc_{dlog}''=(\vec\xi,G_f)$ from the above batch evaluation proof. 
\footnotetext{%
Here, $\vec\eta\cdot\delta_k$ denotes the vector which is $\vec\eta$ at the position of the current circuit $\mathcal C_k$ in the cross-circuit accumulator, and zero elsewhere.
}
\end{protocol}

The multi-point queries to be proven by the batch evaluation argument are as follows.
\begin{itemize}
\item[-] 
$[\hat w(X)], [\hat z_A(X)], [\hat z_B(X)], [\hat U_1(X)], [h_1(X)]$, $[T_{\vec\eta}(\alpha, X)]$ at $\beta$, as 
well as $[\hat U_1(X)]$ at $g\cdot \beta$,
\item[-]
$[T_{\vec\eta}(X,\beta)]$ at $\alpha$, $[T_{H'}(X,\beta)]$ at $\alpha'$,
and $[T_{H''}(\gamma, Y)]$, $C_T'$ from $acc_\mathcal C'$ at  $\beta$,
\item[-]
$[T_{\vec\eta}(X,\beta)]+\lambda\cdot [T_{H'}(X,\beta)]$ at $\gamma$, and 
$C'$ from $acc_{dlog}'$ at $\gamma$.
\end{itemize}
For the sake of completeness we summarize the batch evaluation argument in Section \ref{s:MultiPointSinglePoint}.

\medskip
The following theorem states that the main recursive argument, i.e. Protocol \ref{p:CompleteArgument} extended by the predicate check on $acc''_{\mathcal C}$, is a zero-knowledge argument of knowledge.
We point out that we use the Halevi-Micali \cite{PoKHaleviMikali} notion of proof of knowledge for negligible soundness error, see Definition \ref{def:KnowledgeSoundness} and Defintion \ref{def:ArgumentOfKnowledge}.

\begin{thm}
\label{thm:CompleteProtocol}
If the dlog commitment scheme is computationally binding (Definition \ref{def:Binding}) then Protocol \ref{p:CompleteArgument}, extended by the predicate verification on the resulting inner sumcheck accumulator $acc_\mathcal C''$, is a perfectly honest verifier zero-knowledge (Definition \ref{def:ZeroKnowledge}) argument of knowledge (Definition \ref{def:ArgumentOfKnowledge}) for the relation 
\begin{multline*}
\mathcal R = \big\{((\mathcal C, acc_\mathcal C',acc_{dlog}', x),w) :  (x,w)\in R_{\mathcal C_k}
\wedge \phi(acc_{\mathcal C}')=1
\\
\wedge  \phi_{dlog}(acc_{dlog}')=1 
\wedge (acc_C',acc'_{dlog}) \text{ is consistent with }x
\big\},
\end{multline*}
where $\mathcal C=\{\mathcal C_1,\ldots, \mathcal C_L\}$ is a collection of rank-one constraint systems.
Here, $R_{\mathcal C_k}$ denotes the R1CS relation given by the circuit $\mathcal C_k$, and  $\phi$ and $\phi_{dlog}$ are as in Section \ref{s:InnerSumcheckAggregationGeneral} and Section \ref{s:IPAAggregation}
\end{thm}
The proof of Theorem \ref{thm:CompleteProtocol} is given in Section \ref{s:ProofCompleteArgument}.
In practice we use the Fiat-Shamir transform to turn Protocol \ref{p:CompleteArgument} into a non-interactive argument of knowledge which is zero-knowledge against arbitrary polynomial time adversaries.

\section{A note on performance}
\label{s:Performance}

Inner sumcheck aggregation is particularly effective when the number of previous accumulators is low, as seen from the operations counts in Table \ref{t:InnerSumcheckAggregation}.
For a single previous accumulator ($\ell=1$) representing the case of linear recursion, the prover effort for the recursive argument is comparable to that of standard Marlin for a circuit of R1CS density $d=1$.
Having $\ell=4$ previous accumulators, as in our Darlin PCD scheme, the equivalent density is about $d=1.5$.

\begin{table}[h!]
\caption{%
Recursion prover with and without inner sumcheck aggregation in terms of FFT operations and multi-scalar multiplications, for in-degree $\ell$ (without opening proof).
}
\label{t:InnerSumcheckAggregation}
\vspace*{3mm}
\centering
\begin{tabular}{|l|c|c|}
\cline{2-3}
\multicolumn{1}{c|}{} & polynomial arith. & commit
\\\hline
 intial round  &  $3 ~\textsf{FFT}(n)$   & $3~\textsf{MSM}(n)$
\\
\multirow{2}{*}{outer sumcheck
}  
	&  $2 ~\textsf{FFT}(n) + 2~\textsf{FFT}(2n)$
		& \multirow{2}{*}{$2 ~\textsf{MSM}(n) + 1~\textsf{MSM}(2n)$}
\\
	& $+ 3~\textsf{FFT}(3n)$
		&
\\
aggregation rounds 
	& $(4+\ell)~\textsf{FFT}(n)$
		& $(2+\ell)~\textsf{MSM}(n)$ 

\\\hline
overall 
	 &  $\approx  (15+ \ell)\cdot ~\textsf{FFT}(n)$ 
		& $\approx (9+\ell)~\textsf{MSM}(n)$
\\\hline\hline
without aggregation
	 &  $\approx  (15 + 5\cdot d)~\textsf{FFT}(n)$ 
		& $\approx (7+4\cdot d)~\textsf{MSM}(n)$
\\\hline
\end{tabular}
\end{table}

Compared to a standard Marlin prover for circuits with density $d=2$ the performance improvement is estimated at $27\%$, as indicated by our estimates from Table \ref{t:DarlinVsMarlin} in Section \ref{s:Introduction}.
The timing estimates from this table are based on a detailed simulation of a Darlin prover (in terms of MSM, FFT, vector and vector-matrix operations), run on an Amazon EC2 G4dn instance (with 4 Intel Xeon@2.5 GHz and 1 NVIDIA T4) currently offered at a rate of $0.526$ USD per hour.
The number of constraints for verifying two previous proofs stem from detailed paper-and-pencil counts, where our circuit design follows the `deferred arithmetics' technique from \cite{Halo}, which postpones non-native arithmetic checks to the `next' circuit in recursion, in which these operations are again native. (We moreover apply their endomorphism-based scalar multiplication which reduces the number of constraints significantly.)
We vary over different segment sizes for the dlog commitment scheme (cf. Section \ref{s:Segmentation} on segmentation of homomorphic polynomial commitment schemes) starting with the smallest possible domain size $|H|=2^{19}$ to cover the two verifier, and then reducing the committer key to $2^{18}$ and $2^{17}$.
Consequently, the prover times decrease at the cost of increasing proof sizes and the number of constraints for the verifier circuit.


\chapter{Future work}

We will implement Darlin as the recursive main argument of our upcoming Darlin proof carrying data suite \cite{DarlinFullProtocol}, using a $2$-cycle of ordinary elliptic curves such as the Pasta curves \cite{PastaCurves}.
The full suite will cover pure proof merging nodes (for in-degree $1$ and $2$) as well as special purpose nodes with additional consensus specific logic. 
Beyond that a separate transformation chain of arguments for converting Darlin proofs into ordinary Marlin proofs will be provided.
A formal description, including an in-depth security analysis will be given in \cite{DarlinFullProtocol}.

\chapter{Acknowledgements}

The first author is indebted to Maus and Bowie for their appreciated feedback.
Without them, the main recursive argument would miss its most important feature, the whisker feedback loop in the cross-meal aggregation of fish, chicken and beef. 
One of the first readers is also grateful to Peperita, that helped moving away from pairings in exchange for tasty kibble.


\bibliographystyle{alpha}
\bibliography{bibfileSNARKs}

\appendix
\chapter{Appendix}
\label{s:Appendix}



\section{Notation}

We denote the security parameter by $\lambda$, where we throughout consider it in unary representation. 
A function $f(\lambda)$ is \textit{negligible} if for every polynomial $p(\lambda)$, it holds that $\lim_{\lambda\rightarrow\infty}f(\lambda)\cdot p(\lambda)= 0$, or in short $f(\lambda)=o(\nicefrac{1}{p(\lambda)})$. 

Probabilistic algorithms are denoted by capital letters $\mathsf A, \mathsf B$, etc., and we write $y\leftarrow \mathsf A(x)$ if an algorithm $\mathsf A$ ouputs a string $y$ given an input string $x$ while using some internal random coins $r$ uniformly sampled from $\{0,1\}^*$.
Whenever we need to refer to the used random coins $r$, we shall explicitly write $y=\mathsf A(x;r)$. 
We say that $\mathsf A$ is \textit{probabilistic polynomial time (p.p.t.)}, if its run time $T_{x,r}$ on input $x$ and internal random coins $r$ is bounded by some fixed polynomial $p(|x|)$ independent of the random coins, where $|x|$ denotes the length of its input.
We say that $\mathsf A$ is \textit{expected polynomial time} if the expected run time $E(T_{x,r})$, where the expectation is taken over all random coins $r$, is bounded by some polynomial in the length of the input.
The interaction of two interactive probabilistic algorithms $\mathsf A$ and $\mathsf B$ is denoted $\langle \mathsf A, \mathsf B\rangle$, where we explicitly clarify what are the inputs and outputs of both algorithms.


\section{Interactive arguments}
\label{s:ArgumentSystems}



Let $\mathcal R$ be a polynomial time decidable binary relation.
An interactive argument system for $\mathcal R$ consists of three probabilistic polynomial time algorithms 
\[
(\setup, \prove, \verify).
\]
Given the security parameter $\lambda$ in unary representation, $\setup(\lambda)$ outputs a common reference string $crs$ which supports all statement-witness pairs $(x,w)$ up to a certain maximum length $N=N(\lambda)$, which we write in short $(x,w)\in\mathcal R_N$.
Given $(x,w)\in\mathcal R_N$, the algorithms $\prove$ and $\verify$ are used to interactively reason about whether $x$ belongs to the language defined by $\mathcal R$ or not.
We denote their interaction by $tr\leftarrow\langle \prove(x,w), \verify(x)\rangle$ with $tr$ as the transcript of the interaction, and we assume that both algorithms have access to the $crs$ without explicitly declaring them as inputs.
After at most polynomially many steps the verifier accepts or rejects, and we say that $tr$ is accepting or rejecting.

%
%

\begin{defn}[Perfect completeness]
\label{def:Completeness}
An interactive argument system $(\setup$, $\prove,\verify)$ satisfies  \textit{perfect completeness} if
\begin{equation*}
\prob{
\begin{minipage}{3cm}
$\langle\prove(x,w),\verify(x)\rangle$ \text{ is accepting } 
\end{minipage}
\:\left|\: 
\begin{minipage}{3.2cm}
	$crs\leftarrow\setup(\lambda)$, 
	\\
	$(x,w)\leftarrow\mathcal A(\lambda)$, with $(x,w)\in\mathcal R_N$
\end{minipage}
\right.
} 
= 1.
\end{equation*}
\end{defn}



We define knowledge-soundness in the style of \cite{PoKHaleviMikali}.
However we do not dwell on the structure or the message distribution of the blackbox extractor. 
The reason for this choice of definition is the modularity of our proof of Theorem \ref{thm:CompleteProtocol}, which refers to the security result on the batch evaluation argument from \cite{HaloInfinite}.
\begin{defn}[Knowledge-soundness]
\label{def:KnowledgeSoundness}
An interactive argument system $(\setup$, $\prove,\verify)$ for the relation $\mathcal R$ is \textit{knowledge sound} if for every $x$ from $\mathcal L_\mathcal R$ and every adversary $\mathcal A$ which makes $\langle\mathcal A,\verify(x)\rangle$ accept with non-negligible probability $\varepsilon(x) >\negl[\lambda]$, there is a strict polynomial time algorithm $\mathcal E=\mathcal E^\mathcal A$ with blackbox access to $\mathcal A$ which does at most $\poly[|x|,\lambda]/\poly[\varepsilon]$ calls, and overwhelmingly outputs a witness $w$ such that  $(x,w)\in\mathcal R$.
\end{defn}

\begin{defn}
\label{def:ArgumentOfKnowledge}
We say that an interactive argument system $(\setup, \prove,\verify)$ is an \textit{argument of knowledge}, if it is perfectly complete and knowledge sound as defined above.
It is said to be \textit{succinct}, if the size of the transcript is sublinear in the size of $(x,w)\in\mathcal R$.
\end{defn}

\medskip
As we do not require any trust assumptions for the setup, our definition of zero-knowledge does not make use of trapdoors. 

\begin{defn}[Perfect honest verifier zero-knowledge]
\label{def:ZeroKnowledge}
An interactive argument system $(\setup, \prove, \verify)$  is \textit{perfect honest verifier zero-knowledge} if there is a p.p.t. algorithm $\simulate$ 
such that for every p.p.t. algorithm $\mathcal A$, 
\begin{multline*}
    \prob{
        \begin{minipage}{2cm}
            \centering
            $(x,w)  \in\mathcal R_N$ 
            \\
            $\wedge$
            \\
            $\mathcal \mathcal A(tr') = 1$ 
        \end{minipage}
        \:\left|\: 
        \begin{minipage}{2.7cm}     
        	$crs\leftarrow\setup(\lambda)$, 
	        \\
	        $(x,w)\leftarrow \mathcal A(crs)$,
	        \\
	        $tr'\leftarrow\simulate(crs,x)$
        \end{minipage}
        \right.
    }
    \\
    =
    \prob{
        \begin{minipage}{2cm}
            \centering
            $(x,w)  \in\mathcal R_N$ 
            \\
            $\wedge$
            \\
            $\mathcal A(tr) = 1$ 
        \end{minipage}
        \:\left|\: 
        \begin{minipage}{4cm}     
        	$crs\leftarrow\setup(\lambda)$, 
	        \\
	        $(x,w)\leftarrow \mathcal A(crs)$,
	        \\
	        $tr\leftarrow\langle\prove(x,w),\verify(x)\rangle$
        \end{minipage}
        \right.
    }.
\end{multline*}
\end{defn}

\section{Forking Lemmas}

We use the forking Lemma from \cite{BootleGroth} and we obtain strict polynomial time of the sampling algorithm by truncation. 
Assume that $(\setup$, $\prove,\verify)$ is an $(2r+1)$-move public-coin argument, by which we mean that in each round the verifier messages are chosen uniformly at random from a sample space $S$.
Given a transcript $tr$ resulting from the interaction of $\mathcal A(\:.\:)$ with the verifier $\verify(crs,x)$, we denote by $tr|_{\leq i}$, with $i=0,\ldots,r$, the partial transcript consisting of the prover and verifier messages of the first $2i+1$ moves. 
An \textit{$(n_1,\ldots,n_r)$-tree of accepting transcripts} is a tree of depth $r$ which is rooted in a prover's first message $tr|_{\leq 0}$ and in which each node at level  $i\in\{0,\ldots,r-1\}$ represents a partial transcript $tr|_{\leq i-1}$ and has exactly $n_i$ children nodes extending this transcript. 
The tree has overall $K(\lambda)=\prod_{i=1}^r n_i$ leafs standing for complete transcripts in which the verifier eventually accepts.
We assume that the size of $S$ grows superpolynomially in $\lambda$, so that 
\[
\prob{x_1\neq x_2  | x_1,x_2\sample S} > 1-\negl[\lambda].
\]
\begin{lem}[\cite{BootleGroth}]
\label{lem:ForkingLemma}
Let $(\setup,\prove,\verify)$ be a $(2r+1)$-move public-coin interactive proof, and $\mathcal A$ a p.p.t. adversary which runs in expected time $t_\mathcal A$ and succeeds $\langle\mathcal A(\:.\:),\verify(crs,x)\rangle$ with non-negligible probability $\varepsilon= \varepsilon(x)$ on public input $x$.  
If $n_1,\ldots,n_r\geq 2$ are such that $K(\lambda)=\prod_{i=1}^r n_i$ is polynomially bounded, then there exists a p.p.t. algorithm $\mathcal T$ that calls the next message function of $\mathcal A$ at most $2\cdot K(\lambda)/\varepsilon$ times and with non-negligible probability $\varepsilon/2$ outputs an $(n_1,\ldots,n_r)$-tree of accepting transcripts in which all pairs of sibling-node challenges $x_1,x_2$ are subject to $x_1\neq x_2$.
\end{lem}
 


For the sake of completeness we shortly sketch the construction of $\mathcal T$ as claimed in Lemma \ref{lem:ForkingLemma}.
The tree finder algorithm $\mathcal T'$ from \cite{BootleGroth} is a rejection sampler which is allowed to fail at every first completion of $\langle\mathcal A,\verify(crs,x)\rangle$ on a given partial transcript $tr|_{\leq i}$.
It succeeds with probability 
\begin{align*}
p_{tr|_{\leq i}} &= \prob{ \mathcal T' ~\text{completes subtree for } tr|_{\leq i}} 
\\
&= \prob{\langle\mathcal A(\:.\:),\verify(crs, x)\rangle ~\text{succeeds on } tr|_{\leq i} }.
\end{align*}
and runs in expected polynomial time
\[
\expect{\#\text{ of }\mathcal A \text{ calls given } tr|_{\leq i}}  \leq n_i\cdot n_{i+1}\cdot\ldots\cdot n_r.
\]
Overall $\mathcal T'$ is of expected polynomial time calling  $\mathcal A$ at most $K(\lambda)$ times on average, and $\mathcal T'$ succeeds in producing a complete $(n_1,...,n_r)$-tree of accepting transcript with probability $\varepsilon$. 
The probability that such a complete tree of accepting transcripts has collisions (i.e. two sibling challenges coincide) is negligible, see the full version of \cite{BootleGroth}.
Finally limiting the run time of $\mathcal T'$ to $2\cdot K(\lambda)/\varepsilon$ calls of $\mathcal A$ (and returning $\bot$ in that case) yields a strict polynomial time algorithm  which still succeeds with a probability of at least $\varepsilon/2$.



\begin{lem}[\cite{HaloInfinite}]
\label{lem:ConditionalProbability}
Let $\delta$ be such that $0 < \delta \leq \frac{\varepsilon^2}{8 K(\lambda)}$.
Then with probability at least $\varepsilon/4$ the tree finding algorithm $\mathcal T$ from Lemma \ref{lem:ForkingLemma} outputs a tree of accepting transcripts with the following property: 
For every partial transcript $tr|_{\leq i}$ of length $i$ in the tree, the conditional success probability for $\mathcal A(\:.\:)$ continuing the partial transcript $tr|_{\leq i}$ is at least $\delta$.
\end{lem}
\begin{proof}
The tree sampler $\mathcal T$ from Lemma \ref{lem:ForkingLemma} tests at most $2\cdot K(\lambda)/\varepsilon$ partial transcripts $tr|_{\leq i}$. 
Such transcript succeeds with the probability $p_{tr|_{\leq i}}$ as above.
Therefore, the probability that one of these $p_{tr|_{\leq i}}$ is smaller than a given $\delta$ is at most $2\cdot K(\lambda)/\varepsilon \cdot\delta$.
Choosing the latter of at most $\varepsilon/4$ yields the assertion of the lemma.
\end{proof}

We note that the factors $1/2$ in Lemma \ref{lem:ForkingLemma} and $1/8$ in Lemma \ref{lem:ConditionalProbability} are arbitrary. 
Any other choice of these factors $>1-1/\poly[\lambda]$ is possible.


\section{Proof of Theorem \ref{thm:CompleteProtocol}}
\label{s:ProofCompleteArgument}

Theorem \ref{thm:CompleteProtocol} is subject to Protocol \ref{p:CompleteArgument}, extended by the verification of the resulting $acc''=(\gamma,H'',[T_{H''}(\gamma,Y)])$. 
We refer to this extended protocol as the full protocol. 
The batch evaluation proof $\mathsf{Eval}$ is regarded as a subprotocol, and we use knowledge soundness and honest verifier zero-knowledge of it (Theorem \ref{thm:BatchEvaluation} and Theorem \ref{thm:dlog}) to infer the same properties for the full protocol.

\subsection{Knowledge soundness}

Assume maximum degree for the polynomial commitment scheme is $d=\poly[\lambda]$ with $d\geq 2\cdot n + b$.
The proof is divided into two steps. 
In the first one, we special soundness of the algebraic oracle proof.
The second step uses the transcript sampler from Lemma \ref{lem:ForkingLemma} to construct a strict polynomial time extractor from the strict polynomial time extractor of the batch evaluation argument.

\subsubsection{Step 1. Special soundness.} 
Consider the protocol as an interactive oracle proof where the oracles are guaranteed having a degree of at most $d=\poly[\lambda]$.
Besides the arithmetic checks on the evaluation claims, the verifier checks the oracle for $T_{H''}(\gamma,Y)$ by reading it in full length\footnote{%
Equivalently, the verifier may query the oracle at $d+1$ different points and reconstructs the polynomial from the values.
}
and compare it against the polynomial described by $H''$ as computed in the protocol.
We claim that this `algebraic oracle proof' is $(m_1, m_2, m_3, m_4)$-special sound, with
\[
(m_1,m_2,m_3,m_4,m_5)= (3, n, 2\cdot d + 1, 2, d+1),
\]
in the following sense: 
Given an $(m_1,m_2,m_3,m_4,m_5)$-tree of accepting transcripts with pairwise distinct verifier challenges for $\eta\in F,\alpha,\beta\in F\setminus H, \lambda\in F, \gamma\in F$, respectively, then
\begin{itemize}
\item
the polynomial in the oracle from $acc_{dlog}'$ is the claimed reduction polynomial $h(\vec\xi',X)$,
\item 
the oracles intended for $T_{H'}(\gamma',Y)$ and $T_{\eta}(\alpha,Y)$ in fact carry the correct polynomials, and
\item
the polynomial $y(X)= x(X)+(X^\ell -1)\cdot\hat w(X)\bmod (X^n-1)$ with $\hat w(X)$ from $[\hat w(X)]$ satisfies the R1CS identities \eqref{e:QAPc}, \eqref{e:QAPa}, \eqref{e:QAPb}.
\end{itemize}
This is true for the following reasons:
\begin{itemize}
\item
$m_1=3$ pairwise distinct $\eta_1,\eta_2,\eta_3\in F$ are sufficient to derive the R1CS identities \eqref{e:QAPc}, \eqref{e:QAPa}, \eqref{e:QAPb} from the `lincheck' identity \eqref{e:lincheck}.
(The Vandermonde matrix for different choices of $\eta$ is invertible.)

\item
$m_2= n$ pairwise distinct $\alpha_1,\ldots,\alpha_n\in F\setminus H$ allow inverting the reduction of the lincheck identity \eqref{e:lincheck} to the sumcheck identity \eqref{e:Sumcheck} by means of the Lagrange kernel.
(Recall that the sumcheck is obtained from the lincheck by applying $\langle L_n(X,\alpha), \,.\,\rangle$. 
By Lemma \ref{lem:LagrangeKernel} the inner products for any $n$ different values of $\alpha$ allow to uniquely reconstruct the lincheck polynomial modulo $(X^n-1)$.
)

\item
$m_3=2\cdot d+1$ pairwise distinct $\beta_1,\ldots,\beta_{2d+1}\in F\setminus H$ are sufficient 
to infer the  outer sumcheck identity on the full domain $F$, as well as the identity for the first step of the inner sumcheck aggregation.
(The polynomials are of degree at most $d$.)

\item
$m_4=2$ distinct $\lambda_1,\lambda_2\in F$ allow for reconstructing the component polynomials from their linear combination $T_{\vec\eta}(X,\beta) + \lambda\cdot T_{H'}(X,\beta)$ in the second step of the inner sumcheck aggregation.
(Again, since the Vandermonde matrix is invertible.)

\item
and $m_5=d+1$ pairwise distinct $\gamma_1,\ldots,\gamma_{d+1}\in F$ are sufficient to infer both the correctness of the polynomial behind the linear combination $[T_{\vec\eta}(X,\beta)]+\lambda\cdot [T_{H'}(X,\beta)]$, and the polynomial behind the oracle in $acc_{dlog}'$.
(Again, all polynomials are of degree at most $d$.)
\end{itemize}

\subsubsection{Step 2. Extractor.}
Suppose that $\mathcal A$ is a probabilistic polynomial time adversary which succeeds the complete verifier with non-negligible probability $\varepsilon$ on given inputs $(\mathcal C$, $acc_\mathcal C'$, $acc_{dlog}'$, $x)$. 
As $K(\lambda)=m_1\cdot\ldots\cdot m_5$ is polynomial in $\lambda$, Lemma \ref{lem:ForkingLemma} guarantees a strict polynomial time algorithm $\mathcal T$ which calls $\mathcal A$ at most $2\cdot K(\lambda)/\varepsilon = \poly[\lambda]$ times and succeeds with a non-negligible probability of $\varepsilon/2$ in sampling an $(m_1,m_2,m_3,m_4,m_5)$-tree of accepting transcripts as needed for Step 1.
Each partial transcript $tr|_{i\leq 5}$ records the messages until and including the sampling of the last verifier challenge $\gamma$, before entering the batch evaluation protocol $\mathsf{Eval}$. 
By Lemma \ref{lem:ConditionalProbability} we may assume that for each of these partial transcripts $tr|_{i\leq 5}$, the probability that $\mathcal A$ succeeds on it is at least $\delta= \frac{\varepsilon^2}{8\cdot K(\lambda)} > \negl[\lambda]$.
By knowledge-soundness of the batch evaluation argument $\mathsf{Eval}$, there is a strict polynomial time extractor which calls $\mathcal A$ at most $\poly[\lambda]/\poly[\delta]$ times on each of the transcripts $tr|_{i\leq 5}$ 
 and outputs the witness polynomials from $F^{<d+1}[X]$ (including commitment randomnesses) for 
\begin{itemize}
\item
$[\hat w(X)], [\hat z_A(X)], [\hat z_B(X)], [\hat U_1(X)], [h_1(X)]$, $[T_{\vec\eta}(\alpha, X)]$,
\item
$[T_{H'}(\alpha',X)]$ from $acc_\mathcal C'$, as well as $[T_{\vec\eta}(X,\beta)]$, $[T_{H'}(X,\beta)]$, 
\item
$[T_{\vec\eta}(X,\beta)]+\lambda\cdot [T_{H'}(X,\beta)]$, and 
$[h(\vec\xi',X)]$ from $acc_{dlog}'$,
\end{itemize}
of each of the $(\eta,\alpha,\beta,\lambda,\gamma)= (\eta_{i_1},\alpha_{i_2},\beta_{i_3},\lambda_{i_4},\gamma_{i_5})$ in the transcript tree.
These polynomials have values which pass all the verifier checks of the protocol.
Assuming the dlog commitment is computationally binding, the witness polynomials and commitment randomnesses for $[\hat w(X)]$, $[\hat z_A(X)]$, $[\hat z_B(X)]$, $[\hat U_1(X)]$, $[h_1(X)]$, as well as the witness polynomials for the accumulator commitments $[T_{H'}(X,\beta)]$, $[h(\vec\xi',X)]$ overwhelmingly coincide for all the verifier challenges $(\eta_{i_1},\alpha_{i_2},\beta_{i_3},\lambda_{i_4},\gamma_{i_5})$, and for the same reason the polynomials for $[T_{\vec\eta} (\alpha,X)]$ do not depend on $(\beta_{i_3},\lambda_{i_4},\gamma_{i_5})$.
Hence if we replace the commitments by these polynomials we obtain an $(m_1,m_2,m_3,m_4,m_5)$-tree of accepting transcripts as needed for Step 1 to conclude that  $y(X)=x(X)+(X^\ell-1)\cdot\hat w(X) \bmod (X^n-1)$ satisfies the R1CS identities for $\mathcal C_k$, and both $acc_T'$ and $acc_{dlog}'$ are correct. 
The overall run-time of the extractor is strictly bounded by $K(\lambda)\cdot \poly[\lambda]/\poly[\delta] = \poly[\lambda]/\poly[\varepsilon]$ and succeeds with a  non-negligible probability of at least $\varepsilon/4$.
By amplification we obtain the claimed extractor for knowledge-soundness.

\subsection{Zero-knowledge}

Perfect honest verifier zero-knowledge of Protocol \ref{p:CompleteArgument} is an immediate consequence of perfect honest verifier zero-knowledge of the batch evaluation argument (Theorem \ref{thm:BatchEvaluation}) and the same property for the coboundary outer sumcheck.
The latter is obtained from the following auxiliary lemma. 
The proof of it is straightforward, and we leave it to the reader.
\begin{lem}
\label{l:zk}
Assume that $P$ follows Protocol \ref{p:CompleteArgument}.
Then the conditional distribution of $(v_1,v_2,v_3,v_4,v_5,v_6)=(\hat w(\beta), \hat y_A(\beta),\hat y_B(\beta), \hat U(g\cdot \beta), \hat U(\beta),h(\beta))$, conditional to $(\eta, \alpha, \beta)$, is uniform on the relation $\mathcal R_{\eta,\alpha,\beta}$ of all $(v_1,v_2,v_3,v_4, v_5,v_6)$ satisfying the outer sumcheck equation
\begin{multline*}
    T_\eta(\alpha,\beta)\cdot (x(\beta)+(\beta^\ell-1)\cdot v_1) - L_n(\beta,\alpha)\cdot (v_2 +\eta\cdot v_3 + \eta^2\cdot v_2\cdot v_3) 
    \\
    - v_4 + v_5 
    = v_6\cdot (\beta^n-1).
\end{multline*}
\end{lem}



\medskip
Using Lemma \ref{l:zk} the simulator for the outer sumcheck is constructed as follows.
Given a consistent previous accumulator $acc'=(acc'_{\mathcal C}, acc'_{dlog})$ and any circuit input $x\in F^\ell$, it first samples $\eta$ and $(\alpha,\beta)$ uniformly from $F$ and $(F\setminus H)^2$, respectively, and $(v_1, v_2, v_3,v_4,v_5,v_6)$ uniformly from $\mathcal R_{\eta,\alpha,\beta}$ (by choosing $v_1,\ldots, v_5$ uniformly from $F$ and $v_6$ as determined by the outer sumcheck equation), and then crafts arbitrary polynomials $\hat w(X)$, $\hat y_A(X)$, $\hat y_B(X)$, $\hat U(X)$ and $h(X)$ of degree less than $d$ which evaluate at $X=\beta$ (and $\hat U(X)$ also at $X=g\cdot\beta$) to the corresponding values.
All these polynomials are committed using hiding randomnesses, and the aggregation rounds Step (3) and Step (4) of Protocol \ref{p:CompleteArgument} are performed as in an honest prover-verifier interaction. 
Since the dlog commitment scheme is perfectly hiding, the resulting conversation is identically distributed as in an ordinary honest prover-verifier interaction for Step (1-4).
These conversations are completed by calling the simulator for the batch evaluation argument on the collected commitments.
Since the latter partial transcripts are identically distributed to an honest prover-verifier interaction, so are the completed conversations.


\section{Proof of Theorem \ref{thm:CoboundaryMarlin}}
\label{s:ProofCoboundaryMarlin}

The proof of Theorem \ref{thm:CoboundaryMarlin} is almost identical to that of Theorem \ref{thm:CompleteProtocol}, hence we only point out the differences. 
We assume that the indexer polynomial $row_M(X)$, $col_M(X)$, $row.col_M(X)$ and $val.row.col_M(X)$, $M=A,B,C$, as defined in Section \ref{s:Marlin} are verified against their commitments in a precomputation phase of the protocol.

\medskip
For knowledge-soundness, observe that Coboundary Marlin viewed as an algebraic oracle proof, is $(m_1$,$m_2$,$m_3$,$m_4)$-special sound with
\[
(m_1,m_2,m_3,m_4) = (3, n, 2\cdot d+1,4\cdot m - 3).
\]
Here, $d=\poly[\lambda]$ is the maximum degree of the oracle polynomials, $n$ and $m$ are the sizes of the domains $H$ and $K$, and $m_1,m_2,m_3,m_4$ correspond to the verifier challenges $\eta\in F$, $\alpha, \beta, \gamma \in F\setminus H$, respectively.
The reasoning for $m_1$, $m_2$, $m_3$ is as befor, and $m_2=4\cdot m - 2$ by the degree of the inner sumcheck identity \eqref{e:InnerSumcheck}.  
Based on these modified soundness numbers, the strict polynomial time extractor is constructed as in Section \ref{s:ProofCompleteArgument} by using the forking lemma (Lemma \ref{lem:ForkingLemma}) and the extractor for the batch evaluation argument.

\medskip
The proof for perfect honest verifier zero-knowledge can be taken over almost verbatim, replacing completing of the simulated transcript for the outer sumcheck by an honest prover-verifier run of the inner sumcheck instead of the aggregation rounds.

\section{Polynomial commitment schemes}
\label{s:PolynomialCommitments}

We regard a polynomial commitment scheme consisting of four probabilistic polynomial time algorithms 
\[
(\setup, \comm, \prove, \verify).
\]
Given the security parameter $\lambda$, $(ck,vk)\leftarrow\setup(\lambda)$ generates a common reference string consisting of a committer key $ck$ and a verifier key $vk$ supporting polynomials over some finite field $F$ having degree of at most $N=N(\lambda)$, where $N$ is polynomial in $\lambda$. 
Given the committer key $ck$, the commitment of a polynomial $p(X)$ of degree at most $N$ is computed by $C=\comm(p(X); r)$, where $r$ denotes the used random coins.
(We again omit $ck$ from the inputs for brevity.)
We regard $(\setup,\prove,\verify)$ as a succinct interactive argument system for the relation
\[
\mathcal R = \big\{
    ((C,x,v),(p(X),r)) ~:~ C=\comm(p(X);r) \wedge p(x) = v
\big\},
\]
and call the polynomial commitment scheme to satisfy completeness, zero-knowledge and witness-extended emulation if the interactive argument system does.
We refer to the interaction $\langle \open,\verify\rangle$ as opening proof for the polynomial $p(X)$ at the point $x\in F$. 

The security notions \textit{computational binding} and \textit{perfect hiding} are as for general non-interactive commitment schemes $(\setup,\comm)$. 
For the sake of brevity, we directly cite them applied to polynomial commitment schemes:

\begin{defn}[Perfect Hiding] 
\label{def:Hiding}
We say that a polynomial commitment scheme $(\setup,\comm,$ $\open,\verify)$ is \textit{perfectly hiding} if every polynomial adversary $\mathcal A$ has no advantage over pure guessing when distinguishing the commitments of two adversarially chosen polynomials,
\begin{equation*}
    \prob{
    \:
    b^\star = b 
    \:\left|\: 
    \begin{minipage}{6.8cm}
	    $(ck,vk)\leftarrow\setup(\lambda)$, 
	    \\
	    for $i=1,2$
	        \\
	        \hspace*{0.5cm}$p_i(X)\leftarrow\mathcal A(ck,vk)$, 
            $\deg (p_i(X))\leq N(\lambda)$,
	        \\
	        \hspace*{0.5cm}$C_i\leftarrow \comm(p_i(X))$
	    \\
	    $b\sample\{0,1\}$,
	    $b^*\leftarrow\mathcal A(C_b,C_{1-b})$
    \end{minipage}
\right.
} 
= \frac{1}{2}.
\end{equation*}
\end{defn}

\begin{defn}[Computational Binding] 
\label{def:Binding}
A polynomial commitment scheme $(\setup,\comm,$ $\open,\verify)$ is \textit{computationally binding} if for every p.p.t. adversary $\mathcal A$,
the probability to find two different messages $(p_i(X),r_i)$, $i=1,2$, having the same commitment is negligible:
\begin{multline*}
    \prob{
    \left.
        \begin{minipage}{5.3cm}
            \centering
            $\comm(p_1(X);r_1) = \comm(p_2(X);r_2)$
            \\
            $\wedge$
            \\
            $(p_1(X),r_1)\neq(p_2(X),r_2)$
        \end{minipage}
    \right| 
    \begin{minipage}{5.3cm}
	    $(ck,vk)\leftarrow\setup(\lambda)$, 
        \\
	    $(r_1,p_1(X),r_2,p_2(X))\leftarrow\mathcal A(ck,vk)$ 
    \end{minipage}
} 
\\
= \negl[\lambda].
\end{multline*}
\end{defn}

We further make use the notion of a homomorphic schemes, again directly applied to polynomial commitment schemes:
\begin{defn}[Homomorphic commitment]
A polynomial commitment scheme $(\setup$, $\comm,\open,\verify)$ with commitments in a \textit{commitment group} $(\mathcal G,+)$ is \textit{homomorphic} (or, \textit{linear}), if
\begin{equation*}
    \comm(p_1(X);r_1) + \comm(p_2(X);r_2) = \comm(p_1(X)+p_2(X);r_1+r_2).
\end{equation*}
\end{defn}

\section{The dlog commitment scheme from \cite{Buenz}}

The dlog polynomial commitment scheme from \cite{Buenz} is an ordinary Pedersen vector commitment. 
Given the coefficient vector $\vec c = (c_i)_{i=0}^{d-1}$ of a polynomial $p(X)$ from $F[X]$ of degree at most $d$, its commitment is the Pedersen linear combination
\begin{align*}
\comm(p(X);r) = r\cdot S + c_0\cdot G_0 + \ldots + c_{d-1}\cdot G_{d-1},
\end{align*}
where $S$ and $(G_i)$ is the committer key, and the optional hiding randomness $r$ is uniformly drawn from $F$.
As in \cite{BootleGroth}, \cite{Bulletproofs}, or \cite{Wahby}, 
the opening proof is an inner product argument, which uses $k=\log(d)$ rounds to gradually reduce the size polynomial in question by one half until ending up with a single-coefficient instance. 
In the course of the reduction the committer key is repeatedly folded into a final committer key $G_f$ for the single-coefficient claim, which depends linearly on the initial committer key by
\[
G_f = h_0\cdot G_0 + \ldots + h_{d-1}\cdot G_{d-1},
\]
where $(h_i)$ are the coefficients of the \textit{reduction polynomial} 
\begin{equation*}
h(\vec\xi, X) = \prod_{i=0}^{k-1} ( 1- \xi_{k-1-i} \cdot X^{2^i}),
\end{equation*}
where $\vec\xi=(\xi_i)_{i=0}^{k-1}$ are the random challenges of the reduction steps.
See \cite{Buenz} for a detailed description.

\begin{thm}[\cite{BootleGroth, Buenz}]
\label{thm:dlog}
Under the assumption that the dlog commitment scheme is computationally binding, the  opening argument
is an honest verifier zero-knowledge (Definition \ref{def:ZeroKnowledge}) argument of knowledge (Definition \ref{def:ArgumentOfKnowledge}).
\end{thm}

Zero-knowledge is proven in \cite{Buenz}.
To obtain a strict polynomial time extractor as demanded by Definition \ref{def:ArgumentOfKnowledge}, one proceeds as in \cite{BootleGroth} to construct an expected polynomial time extractor with average runtime $\mu$ inverse proportional to the success probability $\varepsilon$ of the adversary. 
Truncating it to strict polynomial time $2\cdot \mu/\varepsilon$ zields a running time bounded by $\poly[\lambda]/\poly[\varepsilon]$ calls of $\mathcal A$ and still maintains a success probability of $\varepsilon/2$.   


\section{The batch evaluation protocol from \cite{HaloInfinite}}
\label{s:MultiPointSinglePoint}

We give an informal description of the protocol from \cite{HaloInfinite}, Section $5.2$, for proving a multi-point opening claim
\[
p_i(x_i)=y_i, \quad i=1,\ldots,m,
\]
of given polynomials $p_i(X)$, $i=1,\ldots,m$, for a linear commitment scheme $(\setup, \comm, \open,\verify)$.
The multi-point evaluation claim is equivalent to the following system of algebraic identities over the set $\Omega=\{x_1,\ldots, x_m\}$, 
\[
p_i(X)\cdot L_{m}(x_i,X) - y_i = 0 \mod z(X),\quad i=1,\ldots,m,
\]
where $L_{m}(x_i,X)$ is the Lagrange polynomial for $\Omega$, and $z(X)=\prod_{x\in\Omega} (X-x)$ is the vanishing polynomial of $\Omega$.
A more practical system of identities is achieved if we replace the Lagrange polynomials by their non-normalized variant 
\[
z_i(X) = \frac{z(X)}{X-x_i}= \prod_{x\in\Omega\setminus\{x_i\}}(X-x).
\]
These $m$ identities are reduced to a single identity over $\Omega$ in the usual manner:
The prover receives a random scalar $\rho\sample F$ from the verifier, and shows the combined identity
\[
\sum_{i=1}^m \rho^{i-1} \cdot \left( p_i(X) - y_i\right) \cdot z_i(X)
    = 0 \mod z(X)
\]
instead.
For this, the verifier provides an oracle for the quotient polynomial $q(X)$ in the non-modular identity
\[
\sum_{i=1}^m \rho^{i-1} \cdot \left( p_i(X) - y_i\right)\cdot z_i(X) = q(X)\cdot z(X),
\]
which is probed at a fresh random point $x\sample F$.
This identity is verified by showing that the linear combination $\sum_{i=1}^m \rho^{i-1} \cdot z_i(x) \cdot  p_i(X) -  z(x)\cdot q(X)$ opens at $x$ to the expected value $v = \sum_{i=1}^{m} \rho^{i-1}\cdot y_i\cdot z_i(x)$.

\begin{thm}[\cite{HaloInfinite}, Theorem 6]
\label{thm:BatchEvaluation}
Assume that the commitment scheme is computationally binding.
If the opening argument for the polynomial commitment scheme is a perfect honest verifier zero-knowledge argument of knowledge (Definition \ref{def:ArgumentOfKnowledge}), then the same holds for the batch evaluation protocol.
\end{thm}



\section{Segmentation of linear commitment schemes}
\label{s:Segmentation}

\textit{Segmentation} of a homomorphic polynomial commitment scheme is a useful technique to improve the computational effort of the prover at the cost of increasing the commitment size\footnotemark.
\footnotetext{%
We learned this technique from \cite{Mina} but believe that it is commonly known.
}
One chooses an undersized committer key $ck=(G_0,\ldots,G_{s-1})$, where the \textit{segment size} $s$ is typically magnitudes smaller than the targeted maximum degree $d$, and extends its domain beyond degree $s-1$ by decomposing a polynomial $p(X)$ into
\[
p(X) = p_0(X) + X^s\cdot p_1(X) + \ldots + X^{(k-1)\cdot s} \cdot p_k(X), 
\]
with each $p_i(X)$, $i=1,\ldots, k$, of degree at most $s$. 
(If $d$ is the degree of $p(X)$ then the number $k$ of segment polynomials is equal to $\ceil{(d+1)/s}$.) 
The commitment of $p(X)$ is then defined as the vector of the commitments of its segment polynomials using $ck$, 
\[
\comm(p(X)) := \left(\comm(p_0(X)),\ldots, \comm(p_k(X))\right).
\]
Every evaluation claim $p(x)=v$ of the full-size polynomial is translated to the same claim on the linear combination of its segment polynomials,
\[
LC_x(p_0(X),\ldots,p_k(X)) = 
p_0(X) + x^s\cdot p_1(X) + \ldots + x^{(k-1)\cdot s} \cdot p_k(X),
\]
which is efficiently proven by leveraging the homomorphic property of the scheme.


\section{Facts on the Lagrange kernel}
\label{s:LagrangeKernel}
Let $H=\{x: x^n-1=0\}$ be an order $n$ subgroup of the multiplicative group of a finite field $F$.
The Lagrange kernel 
\[
L_n(X,Y) =\frac{1}{n}\cdot \left( 1+ \sum_{i=1}^{n-1} X^i\cdot Y^{n-i}\right) 
\]
is the unique bivariate symmetric polynomial of individual degree at most $n-1$, such that for $y\in H$ the function $L(X,y)$ restricted to $H$ equals the Lagrange function $L_y(X)$, which evaluates to one at $X=y$, and zero otherwise.
The  kernel has the succinct representation
\[
L_n(X, Y)=\frac{1}{n}\cdot \frac{Y\cdot Z_H(X) - X\cdot Z_H(Y)}{ X- Y},
\]
where $Z_H(X)=X^n-1$ is the vanishing polynomial of $H$.
Lagrange kernels represent point evaluation, as characterized by the following simple Lemma.

\begin{lem}
\label{lem:LagrangeKernel}
Suppose that $H\subset F^*$ is a multiplicative subgroup of order $n$ and $p(X)$ is a polynomial of degree $deg(p(X))\leq n-1$. Then for every $z$ in $F$,
\[
\big\langle L_n(X, z), p(X)\big\rangle_H = \sum_{x\in H} L_n(z, x) \cdot p(x) = p(z).
\] 
\end{lem}
\begin{proof}
Since $p(X) = \sum_{y\in H} p(y)\cdot L_n(X,y)$, it suffices to show the claim for $p(X)=L_n(X,y)$, with $y\in H$.
By the property of $L_n(X,y)$, we have $\big\langle L_n(X, z)$, $L_n(X,y) \big\rangle_H =L_n(y,z)$, which by symmetry is equal to $L_n(X,y)$ at $X=z$.
This completes the proof of the Lemma.
\end{proof}

Marlin \cite{Marlin} uses the generalized derivative 
\[
u_H(X,Y)= \frac{Z_H(X)-Z_H(Y)}{X-Y}
\]
instead of $L_n(X,Y)$ (again as element from $F[X,Y]/\langle X^n-1,Y^n-1\rangle$).
For the generalized derivative a similar inner product formula holds:
For every $z\in F$,
\begin{align*}
\big\langle u_H(X, z), p(X)\big\rangle_H &=  \big\langle u_H(X,X)\cdot L_H(X, z), p(X)\big\rangle_H
\\
&= \big\langle L_n(X, z), u_H(X,X)\cdot p(X)\big\rangle_H = p^*(z),
\end{align*}
where $p^*(X) = u_H(X,X)\cdot p(X) \bmod (X^n-1)$.

\end{document}